\documentclass[reqno]{amsart}
\usepackage{amsmath,amsfonts,amsgen,amstext,amsbsy,amsopn,amsthm, amssymb}

\newtheorem{Lemma}{Lemma}
\newtheorem{Theorem}{THEOREM}

\theoremstyle{definition}
\newtheorem{defn}{DEFINITION}
\theoremstyle{definition}

\theoremstyle{definition}

\newcommand\ord{{\Theta}}

\newcommand{\R}{\mathbb{R}}

\newcommand{\N}{\mathbb{N}}

\newcommand{\be}{\begin{equation}}
\newcommand{\ee}{\end{equation}}
\newcommand{\bea}{\begin{align}}
\newcommand{\eea}{\end{align}}

\newcommand\infspec{{\rm{inf\, spec\,}}}

\newcommand\V{\mathcal{V}}

\DeclareMathOperator{\sgn}{sgn}
\DeclareMathOperator{\tr}{tr} 
\DeclareMathOperator{\const}{const}

\begin{document}
\title{The critical temperature for the BCS equation at weak coupling}
\thanks{\copyright\,2007 by the authors.
This paper may be reproduced, in its entirety, for non-commercial
purposes.}

\author[R.L. Frank]{Rupert L. Frank} \address{Rupert L. Frank,
  Department of Mathematics, Royal Institute of Technology, 100 44
  Stockholm, Sweden. {\it Current address:} Department of Mathematics,
  Princeton University, Princeton NJ 08544, USA }
\email{rlfrank@math.princeton.edu}

\author[C. Hainzl]{Christian Hainzl} 
\address{Christian Hainzl, Departments
  of Mathematics and Physics, UAB,  1300 University
  Blvd, 
Birmingham AL 35294, USA} 
\email{hainzl@math.uab.edu}

\author[S. Naboko]{Serguei Naboko} 
\address{Serguei Naboko, Department of Mathematics, UAB, 1300 University
  Blvd, 
  Birmingham, AL 35294, USA. {\it Permanent address:} Department of Mathematical
  Phy\-sics, St.\ Petersburg University, St.\ Petersburg, Russia}
\email{naboko@math.uab.edu}

\author[R. Seiringer]{Robert Seiringer$^\dagger$}
\thanks{$^\dagger$ R.S. acknowledges partial support by U.S. NSF grant PHY-0353181 and by an
A.P. Sloan Fellowship.}
\address{Robert Seiringer, Department of Physics, Princeton University,
Princeton NJ 08542-0708, USA} 
\email{rseiring@math.princeton.edu}

\date{October 22, 2007}

\begin{abstract}
  For the BCS equation with local two-body interaction $\lambda V(x)$, we
  give a rigorous analysis of the asymptotic behavior of the critical
  temperature as $\lambda \to 0$. We derive necessary and sufficient
  conditions on $V(x)$ for the existence of a non-trivial solution for all
  values of $\lambda>0$.
\end{abstract}

\maketitle

\section{Introduction}

The BCS model has played a prominent role in condensed matter physics
in the fifty years since its introduction \cite{BCS}. Originally
introduced as a model for electrons displaying superconductivity, it
has recently also been used to describe dilute cold gases of fermionic
atoms in the case of weak interactions among the atoms
\cite{Leg,NS,rand,andre,parish,Chen,zwerger}.  We will not be
concerned here with a mathematical justification of the approximations
leading to the BCS model, but rather with an investigation of its
precise predictions.

We consider the BCS equation for a Fermi gas at chemical potential $\mu
\in \R$ and temperature $T>0$, with local two-body interaction
$2\lambda V(x)$. Here, $\lambda>0$ denotes a coupling constant, and
the factor $2$ is introduced for convenience. Because of the many
different applications of the BCS equation, it is important to keep the
discussion as general as possible. Our only assumption on the
interaction potential $V$ will be that it is real-valued and $V\in
L^1(\R^3)\cap L^{3/2}(\R^3)$.

It was shown in \cite{HHSS} that the existence of a non-trivial
solution to the BCS gap equation
\begin{equation}\label{bcse}
\Delta(p) = -\frac \lambda{(2\pi)^{3/2}} \int_{\R^3} \hat V(p-q)
\frac{\Delta(q)}{E(q)} \tanh \frac{E(q)}{2T} \, dq
\end{equation}
with $E(p)= \sqrt{(p^2-\mu)^2 + |\Delta(p)|^2}$ at some temperature
$T>0$ is {\it equivalent} to the fact that a certain {\it linear}
operator has a negative eigenvalue. Here, $\hat V(p)=(2\pi)^{-3/2}
\int V(x) e^{-ipx} dx$ denotes the Fourier transform of $V(x)$.  In
particular, it was shown that this property holds for $T$ less than a
certain {\it critical temperature}, which we denote by $T_c$, whereas
there are no non-trivial solutions to Eq.~\eqref{bcse} for $T\geq T_c$.
According to the usual interpretation of solutions to the
BCS gap equation, the system displays superfluid behavior for all
temperatures $T<T_c$, while it is in a normal phase for $T\geq T_c$.
The analysis in \cite{HHSS} shows that $T_c$ is non-zero for purely
attractive (i.e., non-positive) $V$, and exponentially small in $\lambda$.

The fact that the critical temperature in the {\it non-linear} BCS
equation can be expressed in terms of spectral properties of a linear
operator allows for a more thorough investigation of its properties.
This is the purpose of this paper.  In particular, we shall be
concerned here with the asymptotic behavior of $T_c$ at weak coupling,
i.e., for small $\lambda$.  We shall derive necessary and sufficient
conditions on $V$ for the positivity of $T_c$ for {\it all} $\lambda>0$, as
well as its precise asymptotics as $\lambda\to 0$. The precise
statement of our results is given in Theorem~\ref{asymptofcrittemp}
below.

The linear operator one is led to analyze is of the form
$K_{T,\mu}+\lambda V$ where $K_{T,\mu}$ is a multiplication operator
in momentum space that represents an `effective' kinetic energy. By a
modification of the Birman-Schwinger principle we need to study the
diverging part of the compact operator
\begin{equation}\label{eq:bsop}
(\sgn V) |V|^{1/2} K_{T,\mu}^{-1} |V|^{1/2}
\end{equation}
as $T\to 0$. Note that, if $V$ is not of definite sign, then the latter
operator is not self-adjoint and standard perturbation arguments based on
the variational principle will fail. Still we are able to give a
variational characterization for the leading behavior of the critical
temperature in the weak coupling limit.

Our analysis is somewhat similar in spirit to that of the lowest
eigenvalue of the Schr\"odinger operator $p^2+\lambda V$ in \emph{two}
space dimensions, see~\cite{simon}. This latter case is considerably simpler,
however, since $p^2$ has a unique minimum at $p=0$, whereas
$K_{T,\mu}(p)$ takes its minimal value on the Fermi sphere $p^2=\mu$.
Technically, this is reflected in the fact that the singular part of
the Birman-Schwinger operator $(\sgn V) |V|^{1/2} (p^2+T)^{-1}
|V|^{1/2}$ is of rank one in contrast to that of \eqref{eq:bsop},
which is of infinite rank. In particular, the difficulties stemming
from the non-selfadjointness are not present in the case of $p^2+\lambda
V$.

We would like to emphasize that our approach is not restricted to the
kinetic energy $K_{T,\mu}$ appearing in the BCS model, but can be adopted
to any symbol vanishing on a manifold of codimension one or higher.
Operators of this form arise naturally in various fields of Mathematical
Physics, e.g. in the quantum-mechanical description of particles in a
homogeneous magnetic field or in the analysis of trapped modes in
elasticity theory \cite{F,FW,LSW,Sob}.

\section{Mail results and discussion}

According to the analysis in \cite{HHSS}, the critical temperature
in the BCS model is, in appropriate units, given by the following
expression.

\begin{defn}
For $\mu \in \R$ and $T>0$, let $K_{T,\mu}$ be the multiplication operator in momentum space
$$
K_{T,\mu}(p) = \left( p^2 -\mu\right) \frac { e^{(p^2-\mu)/T}+1
}{e^{(p^2-\mu)/T}-1}\, .
$$
Let $V \in L^1(\R^3) \cap L^1(\R^3)$ be real-valued. The critical
temperature in the BCS model is given by 
\begin{equation}\label{deftc}
T_c(V) = \inf \left\{ T>0\, : \, \infspec \left( K_{T,\mu} + V \right) \geq
0\right\}\,.
\end{equation}
\end{defn}
\medskip

More precisely, it was shown in \cite{HHSS} that Eq.~(\ref{bcse}) has
a non-trivial solution for $T<T_c(\lambda V)$, whereas for $T\geq
T_c(\lambda V)$ it doesn't.  Note that $K_{T,\mu} \geq 2 T$, and that the
essential spectrum of $K_{T,\mu} + V$ is $[2T,\infty)$. Hence, in case
$T_c(V)>0$, it is the largest $T$ such that $K_{T,\mu}+V$ has a zero
eigenvalue.  Note also that $K_{T,\mu}$ becomes $|p^2-\mu|$ as $T\to
0$.

We assume that $\mu>0$ henceforth. For weak potentials $V$, the
critical temperature is determined by the behavior of the potential on
the Fermi sphere $\Omega_\mu$, the sphere in momentum space with
radius $\sqrt{\mu}$.  We denote the Lebesgue measure on $\Omega_\mu$
by $d\omega$.

Let $\V_\mu: \,
L^2(\Omega_\mu) \to L^2(\Omega_\mu)$ be the  self-adjoint operator
\begin{equation}\label{defvm}
\big(\V_\mu u\big)(p) = 
\frac 1{(2\pi)^{3/2}} \frac 1{\sqrt{\mu}}\int_{\Omega_\mu}\hat V(p-q) u(q) \,d\omega(q)\,.
\end{equation}
We note that $\V_\mu$ is non-vanishing if $\hat V(p)$ does not vanish
identically for $|p|\leq 2\sqrt\mu$. Since $V\in L^1(\R^3)$ by
assumption, $\hat V$ is a bounded continuous function, and hence
$\V_\mu$ is a Hilbert-Schmidt operator. It is, in fact, trace class,
as will be shown below, and its trace equals
$\frac{\sqrt{\mu}}{2\pi^2} \int V(x)dx$.

Let $a_{\mu}(V)= \infspec(\V_\mu)$ denote the infimum of the spectrum
of $\V_\mu$.  Since $\V_\mu$ is compact, we have $a_\mu(V)\leq 0$.
Note that, in particular, $a_\mu(V)$ is negative if the trace of
$\V_\mu$ is negative, that is, $a_\mu(V)<0$ if $\hat V(0) =
(2\pi)^{-3/2} \int V(x) dx < 0$. Moreover, by considering a trial
function that is supported on two small sets on the Fermi sphere
separated a distance $|p|$, it is easy to see that $a_\mu(V)<0$ if
$|\hat V(p)|>\hat V(0)$ for some $p$ with $|p|<2\sqrt{\mu}$.

Our main result concerning the critical temperature (\ref{deftc}) is as follows.

\begin{Theorem}\label{asymptofcrittemp}
  Let $V\in L^{3/2}(\R^3)\cap L^1(\R^3)$ be real-valued, and let
  $\lambda>0$.
\begin{itemize}
\item[(i)] Assume that $a_{\mu}(V)<0$. Then $T_c(\lambda V)$ is non-zero for
all $\lambda >0$, and
\begin{equation}\label{crittempformula}
\lim_{\lambda\to 0}   \lambda\, \ln \frac{\mu}{T_c(\lambda V)} =
-\frac{1}{a_{\mu}(V)} \,.
\end{equation}
\item[(ii)] Assume that $a_{\mu}(V)= 0$. If $T_c(\lambda V)$ is
  non-zero, then $\ln (\mu/{T_c(\lambda V)})\geq c \lambda^{-2}$ for some $c>0$ and 
  small $\lambda$.
\item[(iii)] If there exists an $\epsilon>0$ such that $a_{\mu}(V-\epsilon|V|)= 0$,
then $T_c(\lambda V) = 0$ for small enough $\lambda$.
\end{itemize}
\end{Theorem}

Note that Eq.~\eqref{crittempformula} implies that, in case
$a_\mu(V)<0$, the critical temperature has the asymptotic behavior
$$
T_c(\lambda V) \sim \mu e^{ 1/(\lambda a_{\mu}(V))}$$ in the limit of
small $\lambda$. 
%
On the other hand, if $a_\mu(V)=0$ then part (ii) of Theorem~\ref{asymptofcrittemp} implies that
 $T_c(\lambda V)$ is at most as big as
$e^{-\const/\lambda^2}$ for some positive constant. If $a_\mu(V)$
remains zero if $\epsilon |V(x)|$ is subtracted from $V(x)$, then
$T_c(\lambda V)=0$ for small enough $\lambda$, and there is no
superfluid phase at weak coupling.

Although we restrict our attention to local potentials $V$ here, we
remark that a similar analysis can be applied in the case of non-local
potentials as well.

\subsection{Radial Potentials.}
In the special case of radial potentials $V$, depending only on
$|x|$, the spectrum of $\V_\mu$ can be determined more explicitly.
Since $\V_\mu$ commutes with rotations in this case, all its
eigenfunctions are given by spherical harmonics. For $\ell$ a
non-negative integer, the eigenvalues of $\V_\mu$ are then given by
$\frac{\sqrt{\mu}}{2\pi^2} \int V(x) |j_\ell(\sqrt\mu |x|)|^2 dx$,
with $j_\ell$ denoting the spherical Bessel functions. These
eigenvalues are $(2\ell+1)$ fold degenerate.  In particular, we then
have
$$
a_\mu(V) = \inf_{\ell \in \N} \, \frac{\sqrt{\mu}}{2\pi^2} \int V(x)
\left|j_\ell(\sqrt\mu |x|)\right|^2 dx\,.
$$
in the case of radial potentials $V$. 
We remark that
$\sum_{\ell\in\N}(2\ell+1)|j_\ell(r)|^2=1$, hence the expression for
the trace of $\V_\mu$ stated above is recovered.

If $\hat V$ is non-positive, it is easy to see that the infimum is
attained at $\ell=0$.  This follows since the lowest eigenfunction can
be chosen non-negative in this case, and is thus not orthogonal to the
constant function. Since $j_0(r)=\sin(r)/r$, this means that $a_\mu(V)
=(2\pi^2 \sqrt\mu)^{-1} \int V(x) \frac{\sin^2(\sqrt{\mu}|x|)}{|x|^2}
dx$ for radial potentials $V$ with non-positive Fourier transform.

In the limit of small $\mu$ we can use the asymptotics $j_\ell(r)
\approx r^\ell/(2\ell+1)!!$ to observe that, in case $\int V(x)dx<0$,
$a_\mu(V) \approx \frac{\sqrt\mu}{2\pi^2} \int V(x) dx$ as $\mu\to 0$.
Note that $(\lambda/4\pi)\int V(x) dx$ is the first Born approximation
to the {\it scattering length} of $2\lambda V$, which we denote by
$a_0$. Thus, replacing $\lambda a_\mu(V)$ by $2\sqrt\mu a_0/\pi$ and
writing $\mu = k_{\rm f}^2$, we arrive at the expression
$T_c\sim e^{\pi/(2k_{\rm f}a_0)}$ for the critical temperature, which
is well established in the physics literature \cite{gorkov,NS,zwerger}.

In the remainder of this paper, we shall give the proof of
Theorem~\ref{asymptofcrittemp}.

\section{Proof of Theorem~\ref{asymptofcrittemp}}

Note that in case $T_c>0$, the essential spectrum of
$K_{T_c,\mu}+\lambda V$ starts at $2T_c>0$, and hence $T_c$ is the
largest $T$ such that $0$ is an eigenvalue of $K_{T,\mu}+\lambda V$ in
this case. Therefore there exists an eigenstate $|\psi\rangle\in
L^2(\R^3)$ such that $K_{T_c,\mu} |\psi\rangle = -\lambda V
|\psi\rangle$. For a (not necessarily sign-definite) potential $V(x)$ let
us use the notation
\begin{equation*}
V(x)^{1/2} = (\sgn V(x)) |V(x)|^{1/2} \,.
\end{equation*}
The Birman-Schwinger principle then implies that $|\varphi\rangle =
V^{1/2} |\psi\rangle$ satisfies $B_{T_c} |\varphi\rangle = -
|\varphi\rangle$, where
\begin{equation}\label{defofbt}
B_T = \lambda V^{1/2}K_{T,\mu}^{-1}|V|^{1/2}\,.
\end{equation}
Conversely, if $B_T|\varphi\rangle=-|\varphi\rangle$ and
$|\psi\rangle=K_{T,\mu}^{-1}|V|^{1/2}|\varphi\rangle$, then
$|\psi\rangle\in L^2(\R^3)$ and $K_{T,\mu}|\psi\rangle=-\lambda
V|\psi\rangle$.  The existence of a zero eigenvalue for $K_{T,\mu}
+\lambda 
V$ is thus equivalent to the fact that $B_T$ has an eigenvalue $-1$.
Note that $B_T$ is not a self-adjoint operator, however. 

With the aid
of the Birman-Schwinger operator $B_T$, we can thus state the
following alternative characterization of the critical temperature
$T_c(\lambda V)$.

\begin{Lemma}\label{l1}
  For any $T>0$, the Birman-Schwinger operator $B_T$ defined in
  (\ref{defofbt}) is Hilbert-Schmidt and has real spectrum. If
  $T_c(\lambda V)>0$, the smallest eigenvalue of $B_{T_c}$ equals
  $-1$. Moreover, in case $T_c(\lambda V)=0$, the spectrum of $B_T$ is
  contained in $(-1,\infty)$ for any $T>0$.
\end{Lemma}

\begin{proof}
  The Hilbert-Schmidt property follows from the
  Hardy-Littlewood-Sobolev inequality \cite[Thm.~4.3]{LL}, using that
  $V\in L^{3/2}(\R^3)$ and that $K_{T,\mu} \geq \const (1+
  p^2)$. Moreover, $B_T$ is the product of a self-adjoint operator
  (multiplication by $\sgn(V(x))$) and a non-negative operator, hence it has
  real spectrum.

  We have already shown above that $-1$ is an eigenvalue of $B_{T_c}$
  in case $T_c>0$. Moreover, because of strict monotonicity of
  $K_{T,\mu}$ in $T$, $-1$ is not an eigenvalue of $B_T$ for all $T >
  T_c$. This implies that $B_{T_c}$ has no eigenvalue less than $-1$,
  for otherwise there would be a $T>T_c$ for which $B_T$ has
  eigenvalue $-1$ since the eigenvalues of $B_T$ depend continuously
  on $T$ and approach $0$ as $T\to \infty$.

  In the same way, one argues that $B_T$ does not have an eigenvalue
  less than or equal to $-1$ if $T_c=0$.
\end{proof}

Let $J$ be the unitary operator that multiplies by $\sgn(V(x))$. To be
precise, we define $\sgn(V(x))=1$ in case $V(x)=0$. Moreover, let $X$
denote the self-adjoint operator on $L^2(\R^3)$ with integral kernel
$$
X(x,y)= |V(x)|^{1/2}  \frac1{2\pi^2} \frac{\sin\sqrt\mu|x-y|}{|x-y|} |V(y)|^{1/2}\,.
$$
We note that the $X$ is a non-negative trace-class operator, with trace $\tr [X]
=\frac{\sqrt\mu}{2\pi^2} \int |V(x)| dx$. Hence also $JX$ is trace-class, and
$\tr [JX] =\frac{\sqrt\mu}{2\pi^2}\int V(x) dx$.
Define $Y_T$ by
\begin{equation}\label{defY}
B_T = \lambda \ln\left(1+\frac\mu{2T}\right) J X + \lambda Y_T \,.
\end{equation}
We have

\begin{Lemma}\label{decomp}
Let $V\in L^1(\R^3)\cap L^{3/2}(\R^3)$.  Then, for any $T>0$, the
operator $Y_T$ defined in (\ref{defY}) is Hilbert-Schmidt, and its
Hilbert-Schmidt norm is bounded
uniformly in $T$, i.e., $\sup_{T>0} \tr[ Y_T^\dagger Y_T]\! <\! \infty$.
\end{Lemma}

The proof of this lemma will be given in the next section.
Lemma~\ref{decomp} shows that the singular part of the operator
$B_T$ as $T\to 0$ is given by $JX$. This observation will 
enable us to recover the exact asymptotics of $T_c(\lambda V)$ as
$\lambda \to 0$.

The operator $JX$ is closely related to $\V_\mu$ defined in
(\ref{defvm}). In fact, the two operators are isospectral.

\begin{Lemma}\label{aequivbvdv}
The spectrum of $JX$ on $L^2(\R^3)$ equals the
spectrum of $\V_\mu$ on $L^2(\Omega_\mu)$. 
\end{Lemma}

\begin{proof}
Let $A: L^2(\R^3)\mapsto L^2(\Omega_\mu)$ denote the operator which
maps $\psi\in L^2(\R^3)$ to the Fourier transform of
$|V|^{1/2}\psi$, restricted to the sphere $\Omega_\mu$. Note that
$|V|^{1/2}\psi\in L^1(\R^3)$ and hence it has a bounded and continuous
Fourier transform. Moreover, let $B: L^2(\Omega_\mu)\mapsto
L^2(\R^3)$ be defined by
$$
(Bu)(x) = V(x)^{1/2} \frac 1{(2\pi)^{3/2}}\frac 1{\sqrt\mu}
\int_{\Omega_\mu} u(p) e^{ipx}
\, d\omega(p)\,.
$$
Using the fact that $\int_{\Omega_\mu} e^{ipx}
d\omega(p)=4\pi\sqrt{\mu}|x|^{-1}\sin \sqrt{\mu}|x|$ it is easy to see
that $JX = BA$, while $AB = \V_\mu$. Hence they have the same
spectrum, except possibly at zero. 
Indeed, if $AB|f\rangle=\lambda |f\rangle$ with $\lambda\neq 0$,  then
$|g\rangle=B|f\rangle \neq 0$ and $BA|g\rangle=\lambda |g\rangle$.
Since both operators are Hilbert-Schmidt operators on
infinite-dimensional spaces, 0 is an element of both spectra.
\end{proof}

We now study the behavior of the spectrum of $JX$ under small
perturbations. We will show that for $\alpha>0$ the spectrum of
\begin{equation}\label{vme}
\alpha J X  + \lambda Y_T
\end{equation}
differs from the spectrum of $\alpha J X$ by at most
$\ord(\sqrt{\alpha\lambda})+ \ord(\lambda)$, uniformly in $T$. Here and in
the following, we use the notation $\ord(t)$ to indicate an expression
that is bounded as $c t\leq \ord(t)\leq C t$ for constants $0<c<C$.

Pick a $z$ that stays away a distance $d$ from the spectrum of $\alpha
JX$. By expanding in a Neumann series, we see that $\alpha J X +
\lambda Y_T-z$ has a bounded inverse provided
$$
\lambda \| Y_T \| \,  \left\| \frac 1 {\alpha JX - z}\right\| < 1\,.
$$
We have
$$
\frac 1{\alpha JX - z} = -\frac 1z +\frac \alpha z  J X^{1/2} \frac {1}{\alpha X^{1/2}J X^{1/2} - z} X^{1/2}\,.
$$
Since $X^{1/2}J X^{1/2}$ is a self-adjoint operator having the same
spectrum as $JX$, we can bound $\|1/(\alpha X^{1/2}J X^{1/2} - z)\|
\leq 1/d$ for any $z$ a distance $d$ away from the spectrum of $\alpha
JX$. We conclude that $\left\|(\alpha JX - z)^{-1}\right\| \leq 1/d+\alpha \| X\|/d^2$.
Hence $z$ is not in the spectrum of (\ref{vme}) if $d \geq
\ord(\sqrt{\lambda \alpha})+\ord(\lambda)$.

Since the spectrum of $\alpha JX + \lambda Y_T$ depends continuously
on $\lambda$, we have thus proved the claim. In particular, it follows
that the lowest eigenvalue of \eqref{vme} equals the lowest eigenvalue
of $\alpha JX$ plus terms that are at most of order $\ord(\sqrt{\alpha\lambda})+
\ord(\lambda)$.

We now have the necessary prerequisites to give the proof of
Theorem~\ref{asymptofcrittemp}.

\begin{proof}[Proof of Part (i)]
According to Lemma~\ref{aequivbvdv}, we have $a_{\mu}(V)=\infspec
JX$. Assume now that $\infspec JX<0$. Since $Y_T$ is bounded uniformly
in $T$, we see that the spectrum of $B_T = \lambda \ln(1+\mu/2T) JX +
\lambda Y_T$ becomes arbitrarily negative for $T\to 0$, and hence
$T_c(\lambda V)>0$ for any $\lambda >0$.
Moreover, $\lambda \ln (1+\mu/2T_c)$ is bounded away from zero as
$\lambda\to 0$.

We have shown above that the lowest eigenvalue of~$B_T$ is bounded
from above and below by $\lambda a_{\mu}(V)\ln(1+\mu/2T) +
\ord(\sqrt{\lambda})$, uniformly in $T$. Since at $T=T_c$ this lowest
eigenvalues equals $-1$, we conclude Eq.~(\ref{crittempformula}).
\end{proof}

\begin{proof}[Proof of Part (ii)]
For $T>0$, let $\alpha = \lambda \ln (1+\mu/2T)$ for simplicity. Under
the assumption that the spectrum of $J X$ is non-negative, the lowest
eigenvalue of $B_T = \alpha J X + \lambda Y_T$
is bigger than $-\ord( \sqrt{\alpha \lambda})$, as shown above. This immediately
implies that $B_T$ can only have an eigenvalue $-1$ if
$\alpha\lambda \geq \ord(1)$, or $\ln(\mu/T) \geq \ord(1/\lambda^2)$ for small
$\lambda$.
\end{proof}

\begin{proof}[Proof of Part (iii)]
Let again $\alpha = \lambda \ln(1+\mu/2T)$, and recall that $B_T = \alpha J X + \lambda Y_T$. 
Since the operator $1+\lambda Y_T$ is
invertible for small enough $\lambda$, we are able to rewrite
$$
1+B_T = (1+\lambda Y_T)\left(1+(1+\lambda
Y_T)^{-1}\alpha JX \right)\,.
$$
Hence $1+B_T$ does not have a zero eigenvalue for any
$\alpha\geq 0$ if the spectrum of
$(1+\lambda Y_T)^{-1} J X$ is non-negative. Note that $JY_T$ is self-adjoint, since $JB_T$ is
self-adjoint and $J^2=1$. Hence $(1+\lambda Y_T)^{-1} J X$  has the same spectrum
as the self-adjoint operator
\begin{equation}\label{aax}
X^{1/2} \frac1{ J + \lambda J Y_T } X^{1/2}\,.
\end{equation}
This operator is non-negative for small $\lambda$ if $X^{1/2}J X^{1/2}
\geq \epsilon X$ for some $\epsilon>0$, since then 
$$
\eqref{aax} = X^{1/2} J X^{1/2} 
- \lambda X^{1/2}
 Y_T  \frac1{ J + \lambda JY_T }X^{1/2}  \geq X 
\left(\epsilon - \lambda  \| Y_T \| \, \| (1 + \lambda Y_T )^{-1}\|\right)\,.
$$

Note that the range of $X$ is dense in the range of $X^{1/2}$, and
hence it is enough to check the inequality $J\geq \epsilon$ on the
range of $X$. Let $|\psi\rangle$ be in the range of $X$, i.e.,
$|\psi\rangle = X |\phi\rangle$ for some $|\phi\rangle\in L^2(\R^3)$.
Then $\langle \psi |J \psi\rangle \geq \epsilon \langle
\psi|\psi\rangle$ is equivalent to the statement that, for
$|\chi\rangle = |V|^{1/2}|\phi\rangle$,
\begin{align*}\nonumber
   &\int_{\Omega_\mu\times \Omega_\mu}
  \overline{\hat \chi(p)} \hat V(p-q) \hat\chi(q) \,d\omega(p)\, d\omega(q) \\
  & \geq \epsilon \int_{\Omega_\mu\times \Omega_\mu}
  \overline{\hat \chi(p)} \widehat{|V|}(p-q) \hat\chi(q)\, d\omega(p)\,
  d\omega(q) \,.
\end{align*}
This, in turn, is equivalent to $a_\mu(V-\epsilon
|V|)=0$. Under this assumption, we have thus shown that, for small
enough $\lambda$, the operator $B_T$ does not have an eigenvalue $-1$,
for arbitrary $T> 0$. Together with Lemma~\ref{l1}, this proves the claim.
\end{proof}

\section{Proof of Lemma~\ref{decomp}}
By scaling we may assume that $\mu=1$, and we set $K_{T,1}=K_T$
for simplicity. The operator $K_{T}$ can be rewritten as
\begin{equation*}
K_{T}(p) = (|p^2-1|+2T) g(|p^2-1|/T)
\end{equation*}
where $g(t)=t(1+e^{-t})/((t+2)(1-e^{-t}))$. The
integral kernel of $K_T^{-1}$ is given by
$$
K_T^{-1}(x,y) = \frac 1{(2\pi)^{3}} \int_{\R^3} \frac{e^{ip(x-y)}}{(|p^2-1|+2T)
g(|p^2-1|/T)} dp\,.
$$
We decompose $K_T(p)$ as $K_T(p)^{-1} = L_T^{(1)}(p) + M_T^{(1)}(p)$,
where $L_T^{(1)}(p) = \theta(\sqrt{2}-|p|) K_T(p)^{-1}$ and
$M_T^{(1)}(p) = \theta(|p|-\sqrt{2}) K_T(p)^{-1}$. 
Since $b=\inf_t g(t)>0$ one has
\begin{equation*}
M_T^{(1)}(p) \leq  \theta(|p|-\sqrt{2}) b^{-1} |p^2-1|^{-1}\,.
\end{equation*}
Using that $V\in L^{3/2}(\R^3)$, we find with the aid of the
Hardy-Littlewood-Sobolev inequality \cite[Thm.~4.3]{LL} that $
\|V^{1/2} M_T^{(1)} |V|^{1/2}\|_2 $ is bounded independently of
$T$. Here, $\|\,\cdot\,\|_2=(\tr |\,\cdot\,|^2)^{1/2}$ denotes the
Hilbert-Schmidt norm.

Note that the integral kernel of $L_T^{(1)}$ is given by
\begin{equation*}
\frac1{2\pi^2} \int_0^{\sqrt{2}} \frac k{(|k^2-1|+2T)
g(|k^2-1|/T) } \frac{\sin k|x-y|}{|x-y|} \,dk \,.
\end{equation*}
We further decompose $L_T^{(1)} = L_T^{(2)} + M_T^{(2)}$, 
where
\begin{equation*}
L_T^{(2)}(x,y)  =
\frac1{2\pi^2} \int_0^{\sqrt{2}} \frac k{(|k^2-1|+2T) } \frac{\sin
  k|x-y|}{|x-y|} \,dk\,.
\end{equation*}
Estimating $|\sin k|x-y|| \leq \sqrt2 |x-y|$ for $k\leq\sqrt 2$ and
changing variables one easily finds
$$
|M_T^{(2)}(x,y)| \leq \frac{\sqrt 2}{2\pi^2} \int_0^{1/T} \frac
1{t+2}\left(\frac1{g(t)}-1\right) \,dt.
$$
This is bounded independently of $T$ since
$1/g(t)-1\sim 2/t$ as $t\to\infty$. Since $V\in L^1(\R^3)$, we
can bound $\|V^{1/2} M_T^{(2)} |V|^{1/2}\|_2 \leq \int |V(x)|dx\, 
\sup_{x,y}|M_T^{(2)}(x,y)|$, and hence we see that also $\|V^{1/2}
M_T^{(2)} |V|^{1/2}\|_2$ is bounded uniformly in $T$.

Finally, we decompose $L_T^{(2)} = L_T^{(3)} + M_T^{(3)}$, 
where
\begin{align*}
  L_T^{(3)}(x,y) & =
  \frac1{2\pi^2} \int_0^{\sqrt{2}} \frac k{(|k^2-1|+2T)} \,dk\,
  \frac{\sin |x-y|}{|x-y|}  
 \\ & \ = \ln\big(1/(2T)+1\big) \frac1{2\pi^2} \frac{\sin |x-y|}{|x-y|}\,.
\end{align*}
Since $|\sin a -\sin b|\leq |a-b|$ one easily sees that 
\begin{equation*}
|M_T^{(3)}(x,y)| \leq \frac1{2\pi^2} \int_0^{\sqrt{2}} \frac k{k + 1
} \,dk\,.
\end{equation*}
Again, since $V\in L^1(\R^3)$, $\|V^{1/2} M_T^{(3)} |V|^{1/2}\|_2$ is
uniformly bounded. This completes the proof.
\hfill\qed



\begin{thebibliography}{10}

\bibitem{andre} N. Andrenacci, A. Perali, P. Pieri, G.C. Strinati,
  {\it   Density-induced BCS to Bose-Einstein crossover},  Phys. Rev. B
  {\bf 60}, 12410 (1999)

\bibitem{BCS} J. Bardeen, L. Cooper, J. Schrieffer, 
{\it Theory of Superconductivity},  
Phys. Rev. {\bf 108}, 1175--1204 
(1957)

\bibitem{zwerger} I. Bloch, J. Dalibard, W. Zwerger, {\it  Many-Body Physics
  with Ultracold Gases}, Preprint arXiv:0704.3011

\bibitem{Chen} Q. Chen, J. Stajic, S. Tan, K. Levin, {\it  BCS--BEC
 crossover: From high temperature superconductors to ultracold
 superfluids},  Phys. Rep. {\bf 412}, 1--88 
(2005)

\bibitem{F} C. F\"orster, {\it Trapped modes for the elastic plate
    with a perturbation of Young's modulus}, Preprint
  arXiv:math-ph/0609032

\bibitem{FW} C. F\"orster, T. Weidl, {\it Trapped modes for an elastic strip with
perturbation of the material properties},
Quart. J. Mech. Appl. Math. {\bf 59}, 
 399--418 (2006)

\bibitem{gorkov} L.P. Gor'kov, T.K. Melik-Barkhudarov, {\it
    Contributions to the theory of superfluidity in an imperfect Fermi
    gas}, Soviet Physics
  JETP {\bf 13}, 1018 (1961)

\bibitem{HHSS} C. Hainzl, E. Hamza, R. Seiringer, J.P. Solovej,
{\it The BCS model for general pair interactions},  Preprint 
arXiv:math-ph/0703086

\bibitem{LSW} A. Laptev, O. Safronov, T. Weidl, {\it Bound State
    Asymptotics for Elliptic Operators with Strongly Degenerated
    Symbols}, in: Nonlinear problems in mathematical physics and
  related topics I, pp. 233--246,
Int. Math. Ser. (N.Y.), Kluwer/Plenum, New York (2002) 

\bibitem{Leg} A.J. Leggett, 
{\it  Diatomic Molecules and Cooper Pairs}, 
in {\it   
    Modern trends in the theory of condensed matter}, A. Pekalski,
  R. Przystawa, eds., Springer (1980)

\bibitem{LL} E. Lieb, M. Loss, {\it Analysis},  American
  Mathematical Society (2001)

\bibitem{NS} P. Nozi\`{e}res, S. Schmitt-Rink,
  {\it Bose Condensation in
  an Attractive Fermion Gas: From Weak to Strong Coupling
  Superconductivity}, J. Low Temp. Phys. {\bf 59}, 195--211
  (1985)

\bibitem{parish} M. Parish, B. Mihaila, E. Timmermans, K. Blagoev, P.
  Littlewood, {\it  BCS-BEC crossover with a finite-range
    interaction}, Phys. Rev. B {\bf 71}, 0645131--0645136 (2005)

\bibitem{rand} M. Randeria, in {\it Bose-Einstein Condensation},
  A. Griffin, D.W. Snoke, S. Stringari, eds., Cambridge University
  Press (1995)

\bibitem{simon} B. Simon, 
{\it The bound state of weakly coupled
    Schr\"odinger operators in one and two dimensions},
  Ann. Phys. {\bf 97}, 279--288 
(1976)

\bibitem{Sob} A.V. Sobolev, {\it Asymptotic behavior of energy levels
    of a quantum particle in a homogeneous magnetic field perturbed by
    an attenuating electric field}. I, Probl. Mat. Anal. {\bf 9},
  67--84 (1984); II, Probl. Mat. Fiz. {\bf 11}, 232--248 (1986)

\end{thebibliography}
\end{document}